\documentclass[doublecolumn,9pt,journal]{IEEEtran}
\usepackage{times}
\usepackage{cite}
\usepackage{color}
\usepackage{epsf}
\usepackage{epsfig}
\usepackage{graphicx}
\usepackage{graphics}
\usepackage{epstopdf}
\usepackage[small]{caption}
\usepackage{amsmath}
\usepackage{amssymb}
\usepackage{amsthm}
\usepackage{amsxtra}
\usepackage[ruled]{algorithm2e}
\usepackage{algorithmic}
\usepackage{enumerate}
\usepackage{multirow}
\usepackage{enumerate}
\usepackage{amssymb}
\usepackage{lipsum}
\usepackage{setspace}
\usepackage{multirow}
\usepackage{subcaption}
\usepackage{url}
\usepackage{colortbl}

\newtheorem{corollary}{Corollary}
\newtheorem{theorem}{\bf Theorem}

\newtheorem{definition}{\bf Definition}
\newtheorem{remark}{Remark}

\usepackage{rotating} 
\usepackage{graphicx} 

\newcommand{\Rmnum}[1]{\expandafter\@slowromancap\romannumeral #1@}

\makeatother
\begin{document}
\title{Grid Influenced Peer-to-Peer Energy Trading}
\author{Wayes~Tushar,~\IEEEmembership{Senior Member,~IEEE,}~Tapan~Kumar~Saha,~\IEEEmembership{Fellow,~IEEE,}~Chau~Yuen,~\IEEEmembership{Senior Member,~IEEE,}~Thomas Morstyn,~\IEEEmembership{Member,~IEEE,}~Nahid-Al-Masood,~\IEEEmembership{Senior Member,~IEEE,}~H.~Vincent~Poor,~\IEEEmembership{Fellow,~IEEE} and Richard Bean
\thanks{W. Tushar and T. K. Saha are with the School of Information Technology and Electrical Engineering of the University of Queensland, Brisbane, QLD 4072, Australia (e-mail: wayes.tushar.t@ieee.org; saha@itee.uq.edu.au).}
\thanks{C. Yuen is with the Engineering Product Development Pillar of the Singapore University of Technology and Design (SUTD), 8 Somapah Road, Singapore 487372. (e-mail: yuenchau@sutd.edu.sg).}
\thanks{T. Morstyn is with the Oxford Martin School of the University of Oxford, Oxford, UK. (e-mail: thomas.morstyn@eng.ox.ac.uk)}
\thanks{N. Masood is with the Department of EEE, Bangladesh University of Engineering and Technology, Dhaka, Bangladesh. (e-mail: nahid@eee.buet.ac.bd)}
\thanks{H. V. Poor is with the Department of Electrical Engineering at Princeton University, Princeton, NJ 08544, USA. (e-mail: poor@princeton.edu).}
\thanks{R. Bean is with the Redback Technologies, Indooroolippy, QLD 4068, Australia. (e-mail: richard@redbacktech.com).}
\thanks{This work is supported in part by the Queensland Government under the Advance Queensland Research Fellowship AQRF11016-17RD2, in part by NSFC 61750110529, and in part by the U.S. National Science Foundation under Grant ECCS-1824710.}
}
\IEEEoverridecommandlockouts
\maketitle
\begin{abstract}This paper proposes a peer-to-peer (P2P) energy trading scheme that can help the centralized power system to reduce the total electricity demand of its customers at the peak hour. To do so, a cooperative Stackelberg game is formulated, in which the centralized power system acts as the leader that needs to decide on a price at the peak demand period to incentivize prosumers to not seeking any energy from it. The prosumers, on the other hand, act as followers and respond to the leader's decision by forming suitable coalitions with neighboring prosumers in order to participate in P2P energy trading to meet their energy demand. The properties of the proposed Stackelberg game are studied. It is shown that the game has a unique and stable Stackelberg equilibrium, as a result of the stability of prosumers' coalitions. At the equilibrium, the leader chooses its strategy using a derived closed-form expression, while the prosumers choose their equilibrium coalition structure. An algorithm is proposed that enables the centralized power system and the prosumers to reach the equilibrium solution. Numerical case studies demonstrate the beneficial properties of the proposed scheme. 
\end{abstract}
 \setcounter{page}{1}
  \section*{Nomenclature}
\addcontentsline{toc}{section}{Nomenclature}
\begin{IEEEdescription}[\IEEEsetlabelwidth{$V_1,~~V_2,$}]
\item[$N$] Total number of participating prosumers.
\item[$T$] Total number of time slots.
\item[$\mathcal{S}$] Set of seller prosumers.
\item[$\mathcal{B}$] Set of buyer prosumers.
\item[$\mathcal{B}_a$] Set of buyer prosumers who trade with auction price.
\item[$\mathcal{S}_a$] Set of seller prosumers who trade with auction price.
\item[$\mathcal{N}$] Set of participating prosumers.
\item[$n$] Index of each prosumer.
\item[$t$] Index of time slot.
\item[$e_{n,d}(t)$] Energy demand of prosumer $n$ at $t$.
\item[$e_{n,g}(t)$] Energy that $n$ buys from CPS at $t$.
\item[$e_{n,p}(t)$] Energy that $n$ buys from another peer at $t$.
\item[$E_{n,b}(t)$] Bidding energy by buying prosumer $n$ at $t$.
\item[$E_{n,s}(t)$] Bidding energy by selling prosumer $n$ at $t$.
\item[$E_S(t)$] Total energy demand by all customers at $t$.
\item[$E_D(t)$] Total energy demand by prosumers at $t$.
\item[$E_O(t)$] Total energy demand by customers other than prosumers at $t$.
\item[$E_T(t)$] Threshold set by the CPS at $t$.
\item[$E_G(t)$] Energy supply capacity of the CPS at $t$.
\item[$p_{g,b}(t)$] Buying price per unit of energy set by the CPS at $t$.
\item[$p_{g,s}(t)$] Selling price per unit of energy set by the CPS at $t$.
\item[$p_{p2p}(t)$] P2P trading price at $t$.
\item[$p_{n,b}(t)$] Bidding price of $n$ at $t$.
\item[$p_\text{auc}$] Auction price.
\item[$p_\text{max}$] Highest reservation price.
\item[$p_\text{FiT}$] Feed-in-tariff price.
\item[$p_\text{mid}$] Mid-market price.
\item[$U_{n,s}(t)$] Utility of $n$ for selling energy at $t$.
\item[$U_{n,b}(t)$] Utility of $n$ for buying energy at $t$.
\item[$J_{c}(t)$] Cost to the CPS at $t$.
\item[$a, b$] Design parameters for the CPS.
\item[$\alpha_n(t)$] Preference parameter of $n$ at $t$.
\item[$\beta$] Design parameter.
\item[$(\cdot^*)$] Equilibrium value of the parameter $(\cdot)$.
\item[$\Gamma$] Stackelberg game.
\item[$\mathbb{D}_{hp}$] Stability index of $\Gamma$.
\item[$\nu$] Value of a coalition.
\item[$\eta_n$] Burden to prosumer $n$.
\end{IEEEdescription}
\section{Introduction}\label{sec:introduction} Recent advancements in cryptocurrencies and blockchain have led to a proliferation of peer-to-peer (P2P) energy trading schemes~\cite{Kang_TII_2017,Wang_PPR1_2019,Wang_Paper2_2018}. Essentially, P2P trading is a next-generation energy management mechanism for the smart grid that enables a customer of the network to independently participate in energy trading with other prosumers and the grid~\cite{Tushar_SPM_July_2018}. Potential benefits of P2P energy trading include renewable energy usage maximization, electricity cost reduction, peak load shaving, prosumer empowerment, and network operation and investment cost minimization.

To realize some of the above-mentioned benefits, a number of studies have been conducted in P2P trading in recent years. These studies can be divided into three general categories. The first category of studies, such as \cite{Mengelkamp_AppliedEnergy_2017,Li_TII_Aug_2018,YueZhou_AE_July_2018,Sorin_PWRS_2019,Luth_AE_Nov_2018,ChaoLong_AE_Sep_2018,ZhangChenghua_EA_June_2018,Morstyn_TSG_2019}, focus on the financial modeling of P2P energy trading market. In \cite{Mengelkamp_AppliedEnergy_2017}, the authors present the concept of a blockchain-based microgrid energy market without the need for central intermediaries. Blockchain-based energy trading to reduce the cost of establishment in energy limited Industrial Internet of Things nodes is also discussed in \cite{Li_TII_Aug_2018}. A multi-agent simulation framework and a consensus-based approach for P2P energy trading in a microgrid are used in \cite{YueZhou_AE_July_2018} and \cite{Sorin_PWRS_2019} respectively. The role of battery flexibility on P2P trading is discussed in \cite{Luth_AE_Nov_2018} and their controls for P2P energy sharing using a two-stage aggregated battery control are presented in \cite{ChaoLong_AE_Sep_2018}. Finally, game theory based contracts between sellers and buyers of the P2P market are developed in \cite{ZhangChenghua_EA_June_2018} and \cite{Morstyn_TSG_2019}.

The second category of studies focuses on the challenges of transferring energy over the physical layer~\cite{Mengelkamp_AppliedEnergy_2017} of the distribution network. For example, \cite{Chapman_TSG_2018} proposes a methodology to assess the impact of P2P transactions on the network and to guarantee an exchange of energy that does not violate network constraints. A real-time attribution of power losses to each P2P transaction involving one generator and one load node is done by defining some suitable indices in \cite{Zizzo_TII_EA_2018}. A similar study of energy blockchain in microgrids can be found in \cite{Silvestre_Nov_2018}. In \cite{Nikolaidis_TPS_2018}, the authors propose a graph-based loss allocation framework for P2P market in unbalanced radial distribution networks. Finally, \cite{Werth_TSG_2018} shows how to achieve decentralization using P2P frameworks as underlying control structures, and then, implement a pure P2P to eliminate single points of failure.

The final category of studies deals with engaging prosumers in P2P energy trading. Examples of such studies include \cite{Morstyn_PWRS_2018}, \cite{Morstyn_NatureEnergy_2018} and \cite{Tushar_Access_Oct_2018}. A new concept of energy classes, allowing energy to be treated as a heterogeneous product, based on attributes of its source which are perceived by prosumers to have value is introduced in \cite{Morstyn_PWRS_2018}. The authors in \cite{Morstyn_NatureEnergy_2018} design a P2P energy trading platform to incentivize prosumers to form federated power plants. Finally, a motivational psychology based game theoretic approach is proposed in \cite{Tushar_Access_Oct_2018} to attract prosumers to participate in P2P energy trading. More survey of P2P energy trading schemes can also be found in \cite{Sousa_RSER_Apr_2019} and \cite{Jogunola_Energies_2017}.

As can be seen from the above discussion, studies related to P2P energy trading are plentiful and have significantly contributed to prosumer-focused energy management. However, as identified in \cite{Mengelkamp_AppliedEnergy_2017}, integrating the P2P energy trading mechanism into the current energy policy requires P2P trading to fit and complement the existing energy system. Otherwise, P2P trading markets would be difficult to implement in grid-connected systems. One potential way to establish the suitability of a P2P energy trading scheme to implement in the traditional energy system is to show how such a scheme can benefit the centralized power system (CPS). For example, by effectively supporting the outdated electrical grids at Brooklyn, Brooklyn microgrid  is actively working closely with utilities and is currently on its way to being licensed as an energy retailer~\cite{Mengelkamp_AppliedEnergy_2017}. Nonetheless, such benefits of P2P energy trading to the traditional power system are yet to be established.

To this end, this paper proposes a comprehensive analytical framework of a dynamic price based P2P energy trading scheme that can help the CPS reducing its cost of energy production and supply to the prosumers at the peak demand period. In particular, we propose a cooperative Stackelberg game that comprises a CPS as the leader and the prosumers as followers. At the peak hour, the CPS strategically chooses the selling price per unit of energy such that supply of energy to prosumers (consequently, the cost of excess generation or reserve) reduces to zero. In response to the choice of price made by the CPS, prosumers, as followers of the game, use the double auction to participate in a coalition formation game. The objective of the followers is to form suitable coalitions, based on their submitted bids, and participate in P2P energy trading with the neighboring peers of the same coalition with the purpose of meeting the demand of energy without interacting with the CPS. We analyze various properties of the resulting game. In particular, it is shown that, due to the strategy proof property of the auction and stability of the formed coalitions, the proposed game possesses a unique cooperative Stackelberg equilibrium. We derive a closed-form pricing function for the CPS and propose an algorithm that the CPS and prosumers can use to reach the solution of the game. Further, we show that the resultant solution is also prosumer-centric. Using numerical simulation, we assess the properties of the proposed scheme.

We stress that the Stackelberg game has been extensively used in the literature such as in \cite{Tushar-TIE:2015,Sabita_TSG_2013,Wayes-J-TSG:2012} and \cite{Wei_AE_2017} for designing energy trading schemes for the smart grid. However, in these studies, the framework of a noncooperative Stackelberg game is used, in which followers, in response to the leader's strategy, participate in a noncooperative game such as a Nash game, Generalized Nash game, or best response strategies. In this paper, on the other hand, followers' responses are captured via a coalition formation game considering its capacity to dynamically consider the impact of environmental changes on  prosumer's decision on cooperative interaction with one another for P2P trading. Thus, the proposed game is a cooperative Stackelberg game. Consequently, the studied properties, the design of the algorithm, and the resulting solutions are substantially different than the existing studies. We further note that price-based coordination of distributed energy resources to support the grid is also a topic of interest in the literature related to virtual power plant (VPP). However, the coordination strategies of DERs for VPP and P2P are different from one another \cite{Morstyn_NatureEnergy_2018}.

The rest of the paper is organized as follows. The problem  is formulated in Section~\ref{lab:ProblemFormulation} and an analytical framework of the cooperative Stackelberg game is proposed in Section~\ref{lab:GameTheory}. The properties of the game is studied in Section~\ref{lab:Properties}. We provide some results from numerical case studies in Section~\ref{lab:NumericalSimulation}, which is followed by some concluding remarks in Section~\ref{lab:Conclusion}.
\section{Problem Formulation}\label{lab:ProblemFormulation}
To formulate the problem, we assume an energy network that consists of a CPS and a large number of customers. Among the customers, $N$ are prosumers, where $N=|\mathcal{N}|$ and $\mathcal{N}$ is the set of all prosumers. In this paper, customers refer to energy entities in the network that buy energy from the CPS. Prosumers, on the other hand, have energy producing capability, can participate in P2P trading, and sell its energy to the CPS or any other third party. Both the CPS and prosumers have two-way communication and power flow facilities \cite{Tushar_TSG_May_2016} and are parts of the P2P network through a blockchain based platform~\cite{Mengelkamp_AppliedEnergy_2017}. 

Each prosumer $n\in\mathcal{N}$ is a rational individual, which is equipped with a rooftop solar panel with or without a battery in a grid-connected system. At any time $t$, a prosumer $n$ meets its energy demand $e_{n,d}(t)$ from its own solar system and, if there is any surplus after consumption and storage, it can sell the excess energy either to the grid or to other energy entities within the system at the next time slot. Similarly, if a prosumer determines an deficiency at the end of $t$, it can buy its required energy either from the grid or from other entities within the system at $t+1$.

The CPS, on the other hand, always needs to meet the total demand $E_S(t)$ of its customers. If the total demand $E_S(t)$ of the customers is very high, for example, during peak hours, the network could be overloaded and the CPS might need to start a new generation unit or always maintain  reserve to meet the extra demand of its customers. This increases the cost to the CPS significantly~\cite{Vazquez-TPWRS:2007}. To this end, we assume that the CPS has a contract with $N$ prosumers, in which the CPS may instruct the prosumers to not demand any energy from the CPS at the peak hour when their total demand $E_D(t)$ is higher than a threshold $E_T(t)$. Here, $E_T(t)$ is set to ensure that the total demand $E_S(t)$ from the customer do not exceed the supply $E_G(t)$ from the CPS. Now, $E_S(t)$ can be defined as
\begin{equation}
E_S(t) = E_D(t) + E_O(t),
\end{equation}
where, $E_O(t)$ is the total demand of the customers excluding the prosumers in set $\mathcal{N}$. Now, when $E_D(t)>E_T(t)$, where the value of $E_T(t):E_S(t)\leq E_G(t)$ can be determined through suitable artificial intelligence technique, such as in \cite{Muralitharan_Neurocomputing_2018}, based on the information of $E_O(t)$, $E_D(t)$ and $E_G(t)$, the CPS sends a very high pricing signal to the prosumers under contract to encourage them to manage their own energy among themselves without getting any supply from the CPS at $t$. The objective of the CPS is to make $E_D(t) = 0$ through its choice of price $p_{g,s}(t)$ per unit of energy for the selected prosumers at the respective peak demand time slot. Thus, on the one hand, the CPS meets the demand of its other customers without raising the cost of production and supply of energy significantly. On the other hand, this enables the prosumers to participate in P2P trading among one another to meet their demands of that particular time period.

Here, it is important to note that prosumers can participate in P2P trading at any time, as studied in \cite{Morstyn_TSG_2019,Tushar_Access_Oct_2018}. However, in this paper, the main objective is to show how P2P trading can be used as a tool to help the CPS to reduce its cost of oversupply of energy during periods of peak demand. Therefore, we assume that prosumers only participate in P2P energy trading when the grid price is very high at the peak hours. Otherwise, the prosumers exchange their energy with the CPS as a rooftop solar owner in a traditional energy market. Now, to capture the benefit and cost of energy trading to the prosumers and CPS respectively, we define a suitable utility function for the prosumer and a cost function of the CPS in the following subsections.

\subsection{Utility Function of the Prosumer}We consider that the utility to a prosumer consists of two parts: 1) the utility of energy usage (consumption and supply), which can be captured by a logarithm function $\log(\cdot)$ as proposed in \cite{Yu_IoTJ_Dec_2014}. 2) The revenue and cost of trading with another party (the CPS or another prosumer), which is a function of the traded energy amount and the price per unit of energy. To this end, we propose to use the following two functions \eqref{Utility_selling} and \eqref{Utility_buying} 
\begin{eqnarray}
U_{n,s} (t) = \alpha_n(t) \log_2(1+e_{n,g}(t)+e_{n,p}(t)) \nonumber\\+ p_{g,b}(t) e_{n,g}(t) + p_{p2p}(t)e_{n,p}(t)
\label{Utility_selling}
\end{eqnarray}
and
\begin{eqnarray}
U_{n,b} (t) = \alpha_n(t) \log_2(1+e_{n,g}(t)+e_{n,p}(t)) \nonumber\\- p_{g,s}(t) e_{n,g}(t) - p_{p2p}(t)e_{n,p}(t)
\label{Utility_buying}
\end{eqnarray}
to capture the utilities $U_{n,s} (t)$  and $U_{n,b} (t)$  that the prosumer $n$ obtains when it sells and buys energy respectively. Here, $\alpha_n(t)$ is a preference parameter of the prosumer $n$ that captures the satisfaction level of the prosumer for using a unit of energy at $t$, and $e_{n,g}(t)$ and $e_{n,p}(t)$ are the energy amount that the prosumer $n$ trades with the CPS and other peers respectively. $p_{g,s}(t)$, $p_{g,b}(t)$ and $p_{p2p}(t)$ are the grid's selling price, grid's buying price and the P2P trading price respectively.  Clearly, as the system is defined, when a prosumer is participating in energy trading with the CPS, it does not perform P2P trading with another prosumer, i.e., $e_{n,p}(t) = 0, \text{when}~ e_{n,g}(t)>0$, and vice versa.

Now, when $e_{n,g}(t) >0$, the utility to a prosumer $n$ for buying its energy from the CPS under such conditions can be expressed as 
\begin{eqnarray}
 U_{n,b} (t) = \alpha_n(t) \log_2(1+e_{n,g}(t)) - p_{g,s}(t) e_{n,g}(t).
 \label{utility_cond1}
\end{eqnarray}
From \eqref{utility_cond1}, the utility attains its maximum value when $\frac{\delta U_{n,b}(t)}{\delta e_{n,g}(t)}=0$, and hence 
\begin{eqnarray}
e_{n,g}(t) = \frac{\alpha_n(t)}{p_{g,s}(t)\ln 2} - 1.
\label{energy_cond1}
\end{eqnarray}
Indeed, from \eqref{energy_cond1}, the amount of energy that a prosumer buys from the CPS is affected by the price set by the CPS. Hence, if the price is very high, prosumers may decide to look for alternate venues for their energy demand, which is the main avenues that is explored in this study.  
\subsection{Cost Function of the CPS}\label{sec:CPS_Cost}We assume that there is a net cost $J_c(t)$ to the CPS when it trades its energy with the prosumers. $J_c(t)$ has a cost component and a revenue component. The revenue component $p_{g,s}(t)E_{D}(t)$ is the revenue that the CPS obtains by selling the total energy $E_{D}(t)$ to the prosumers at a price $p_{g,s}(t)$ per unit of energy. Meanwhile, the cost component refers to the equivalent cost that the CPS needs to pay for meeting the demand of prosumers beyond its predefined threshold, for example, for violating the network constraint or for starting an additional generator to meet the excess demand. Given this context, the cost function of the CPS is proposed to be
\begin{eqnarray}
J_c (t) =  a([E_{D}(t) - E_{T}(t)]^+)^2 + b[E_{D}(t) - E_{T}(t)]^+\nonumber\\ - p_{g,s}(t)E_{D}(t),
\label{utility_CPS}
\end{eqnarray}
where $[\cdot]^+ = \max(\cdot,0)$ and $a,~b>0$. In \eqref{utility_CPS}, we note that the cost to CPS only occurs when the CPS needs to meet a total demand beyond its defined threshold. To minimize the cost in such cases, the CPS needs to choose a suitable price that would encourage prosumers to purchase their required energy from other alternative venues, instead of buying from the CPS. As such, the CPS sets a suitable price $p_{g,b}(t)$ when $E_{D}(t)>E_{T}(t)$ with the purpose of minimizing its cost. Thus, by setting  $\frac{\delta J_c (t)}{\delta E_D(t)}=0$, we obtain the price $p_{g,s}(t)$ as
\begin{eqnarray}
p_{g,s}(t) = 2a(E_D(t) - E_T(t)) + b.
\label{Price_Pgt1}
\end{eqnarray}
Now, based on the total demand $E_D(t)$ and the threshold $E_T(t)$, the CPS can choose a suitable price by setting the values of $a$ and $b$ in order to alter the demand of the prosumers for secured and sustainable operation of energy trading.

Given this context, following \cite{Yu_IoTJ_Dec_2014}, replacing the value $e_{n,g}(t)$ from \eqref{energy_cond1} to \eqref{utility_cond1} , and then taking the first derivative of $U_{b,n}(t)$ with respect to $e_{n,g}(t)$, we find the maximum price $p_{n,\text{max}}$ that the prosumer $n$ needs to pay the CPS for purchasing its required energy at the peak hour as
\begin{eqnarray}
p_{n,\text{max}} = \frac{\alpha_n(t)}{\ln2}.
\label{price_threshold}
\end{eqnarray}

From \eqref{Price_Pgt1} and \eqref{price_threshold}, to influence a prosumer $n$ not to buy any energy at the peak hour, the CPS needs to set its price $p_{g,s}(t)$ such that
\begin{eqnarray}
p_{g,s}(t)>\frac{\alpha(t)}{\ln2},
\label{price_threshold2}
\end{eqnarray}
where $\alpha(t) = \max(\alpha_1(t), \alpha_2(t). \hdots, \alpha_N(t))$. Therefore, for a predefined value of $a$, to keep the total demand $E_D$ below $E_T$, the CPS sets its pricing parameters $b$ in \eqref{Price_Pgt1} such that
\begin{eqnarray}
b>\frac{\alpha(t)-2a\ln2(E_D(t)-E_T(t))}{\ln2}.
\label{condition_b}
\end{eqnarray}

Now, once the price is set by the CPS, the prosumers require to respond in a way such that they do not need to rely on the CPS's energy, and at the same time, are capable of meeting their energy demand. To assist the CPS in reducing its peak demand through P2P, application of such a model will potentially be found in the future energy markets that will either engage prosumers with existing solar panels (and batteries) with the CPS under an Energy Performance Contract (EPC) \cite{Zhang_EnergyPolicy_2018} or receive investment from the local government or the CPS to install solar panels (and batteries) at the prosumers' premises to receive such services (for example, see \cite{CONSOART_2019}).

Now, to capture the decision making process of each prosumer, when $E_{D}(t)>E_{T}(t)$, we propose an energy management scheme based on a cooperative Stackelberg game in the next section.
\section{Cooperative Stackelberg Game}\label{lab:GameTheory}
To study the decision making process of the CPS and the participating prosumers in the proposed grid instructed P2P energy trading scheme, we use the framework of a Stackelberg game $\Gamma$, which is formally defined in its strategic form as
\begin{eqnarray}
\Gamma = \{\left(\mathcal{N}\cup\text{CPS}\right), U_{n\in\mathcal{N}}(t), J_c(t), e_{n,p}(t), p_{g,s}(t), p_{p2p}(t)\}.\label{def:game}
\end{eqnarray}
In \eqref{def:game}, $\Gamma$ consists of the following components:
\begin{enumerate}[i)]
\item The prosumers in set $\mathcal{N}$ that act as followers and cooperate with one another to choose their strategies in response to the price set by the CPS at the peak hour.
\item $U_n(t)\in\{U_{n,s}(t), U_{n,b}(t)\}$ is the utility of prosumer $n$ for trading its energy at $t$, either as a seller or a buyer, with other prosumers and the CPS.
\item $J_c(t)$ is the cost to the CPS at $t$.
\item $e_{n,p}(t)$ and $p_{p2p}(t)$ are the strategies, that is traded energy and the price per unit of traded energy in P2P market, of the followers while participating in P2P trading.
\item $p_{g,s}(t)$ is the strategy of the CPS when $E_{D}(t)>E_{T}(t)$.
\end{enumerate}

In the proposed $\Gamma$, as discussed previously, the objective of the CPS is to set a price that reduces its energy trading with the prosumers to zero when $E_{D}(t)>E_{T}(t)$. The aims of the prosumers with energy deficiency, on the other hand, are to buy the necessary energy from prosumers of the community with energy surplus via forming suitable coalitions for P2P energy trading. A suitable solution of the proposed $\Gamma$ is the cooperative Stackelberg equilibrium (CSE). At the CSE, neither the CPS nor any prosumer will have any incentive to choose alternative strategies to improve their benefits in terms of achieved cost and utility respectively.
\begin{definition}
Consider the game $\Gamma$ defined in \eqref{def:game}, where $U_{n,s}(t)$, $U_{n,b}(t)$ and $J_c(t)$ are determined by \eqref{Utility_selling}, \eqref{Utility_buying} and \eqref{utility_CPS} respectively. Now, a set of strategies $(p_{g,s}^*(t), \mathbf{e}^*(t), \mathbf{p}_{p2p}^*(t))$ constitutes a CSE of the proposed $\Gamma$, if 
\begin{equation}
{J_c}({p_{g,s}^*})(t) = 0
\end{equation}and
the followers' strategies ($\mathbf{e}^*, \mathbf{p}_{p2p}^*$) in response to $p_{g,b}^*(t)$ establish a $\mathbb{D}_{hp}$ stable coalition structure.
\label{definition:1}
\end{definition} 
\begin{definition}
A group of coalitions is said to be in $\mathbb{D}_{hp}$ stable if, at a given time slot, no player has an interest to split from one coalition and form another new coalition for a better utility.
\end{definition}
\subsection{Leader's Strategy}
Based on our discussion in Section~\ref{sec:CPS_Cost} and the utility of the CPS considered in \eqref{price_threshold2}, clearly the CSE strategy of the CPS is 
\begin{eqnarray}
p_{g,s}^*(t) = 2a^*(E_D(t) - E_T(t)) + b^*.
\label{eqn:Leader-strategy}
\end{eqnarray} 
Thus, at the peak hour, to influence the prosumers to buy their required energy through participating in P2P trading, the CPS sets $p_{g,s}(t)=p_{g,s}^*(t)$, according to \eqref{eqn:Leader-strategy}.  In \eqref{eqn:Leader-strategy}, the value of $a^*$ and $b^*$ can be obtained either based on a tabular database, in which the various pre-defined values of $a^*$ and $b^*$ are stored for different scenarios based on the historical cases, or by forming and executing an additional optimization technique that provides the values of $a^*$ and $b^*$ for the considered scenario at $t$. Nonetheless, in both cases, the values of $a^*$ and $b^*$ will need to satisfy condition \eqref{condition_b}.

Now, according to Definition~\ref{definition:1},  to reach the CSE solution, the followers need to find their strategies to form a coalition structure in response to $p_{g,s}^*(t)$, which is $\mathbb{D}_{hp}$ stable. Given this context, we formulate a coalition formation game among the prosumers in the next section.
\subsection{Followers' Response}\label{sec:FollowersResponse}In response to the leader's strategy, prosumers, as the followers of $\Gamma$, participate in a coalition formation game to decide on their own strategies. Coalition formation game is a branch of cooperative game, in which the main purpose is to analyze the formation of the coalition structure through players' interaction, and study its properties and its adaptability to the environmental variation~\cite{Saad-coalition:2009}. It can be formally defined as a pair $(\mathcal{N},\nu)$, where $\mathcal{N}$ is the set of all participating players and $\nu$ is the value of coalition. Essentially, $\nu$ is a real number assigned to every coalition of prosumers $\mathcal{S}\subset\mathcal{N}$ for P2P trading. 

In the proposed coalition formation game, a prosumer $n$ can independently choose other prosumers with whom it wants to form coalition to trade energy. This decision is motivated by the individual benefit $U_n$ that a prosumer $n$ can obtain by cooperating with the selected prosumers. Thus, in a coalition formation game, choosing a particular coalition is based upon the individual benefit to the prosumer, rather than the coalition value $\nu$. Of course, a prosumer may also choose not to participate in the P2P trading, rather buy its energy from the CPS despite the price $p_{g,s}^*(t)$. Such behaviour is referred to as the non-cooperative behavior of the prosumer in this paper. Nonetheless, to determine how a prosumer chooses other neighboring prosumers to form a coalition in order to participate in P2P trading, we assume that prosumers adopt a double auction scheme, as described in the following section.
\subsubsection{Double auction between prosumers}\label{sec:doubleauction}A double auction market involves a number of buyers and sellers that interact with one another with their reservation prices and bids to decide 1) the price of P2P trading and 2) how much energy each prosumer can trade with one another. Thus, due to the interactive and independent nature of the decision making process of each participant, a double auction is chosen for the proposed P2P energy trading between the prosumers. 

In general, the proposed auction scheme consists of three key elements:
\begin{itemize}
\item Seller: The prosumers in $\mathcal{S}\subset\mathcal{N}$ that have surplus energy to sell.
\item Buyer: The prosumers in $\mathcal{B}\subset\mathcal{N}$ that need to buy energy to meet their needs.
\item Auctioneer: An auctioneer could either be a third party such as the distribution system operator (DSO) that can participate in the trading through a digital trading platform such as the Elecbay, as proposed in \cite{ZhangChenghua_EA_June_2018}, or an automated and secured information system such as the blockchain~\cite{Kang_TII_2017,Andoni_RSER_2019}, which is able to accurately decide the auction price and trading energy amount following some well defined rules based on the information provided by all prosumers.
\end{itemize}
We assume that the auction between prosumers are conducted following two steps. These steps are explained as follows.

\emph{Step 1- Determination of auction price and participating prosumers:} In this step of the proposed auction process, the auctioneer determines the number of prosumers that intend to form a coalition group to trade their energy based on the outcomes determined by the auction process. To do so, first, at each time slot, each seller $n\in\mathcal{S}$ submits its reservation price $p_{n,s}$ and energy $E_{n,s}$ that it is interested to sell to the auctioneer. Similarly, each buyer $n\in\mathcal{B}$ submits its bidding price $p_{n,b}$ and energy $E_{n,b}$ that it is interested to buy from the auctioneer. Second, once all the ordered information is received by the auctioneer, it arranges the reservation prices in  ascending order and the bidding prices in descending order~\cite{Saad_DoubleAuction_Oct_2011}. That is, $p_{n,s},~\forall n\in\mathcal{S}$ and $p_{n,b},~\forall n\in\mathcal{B}$ are ordered without loss of generality as
\begin{eqnarray}
p_{1,s}<p_{2,s}<\hdots<p_{S,s}~\text{and}~p_{1,b}>p_{2,b}>\hdots> p_{B,b}
\label{eqn:step1_eqn1}
\end{eqnarray}
respectively, where $B=|\mathcal{B}|$ and $S=|\mathcal{S}|$.

\begin{figure}[t!]
\includegraphics[width=\columnwidth]{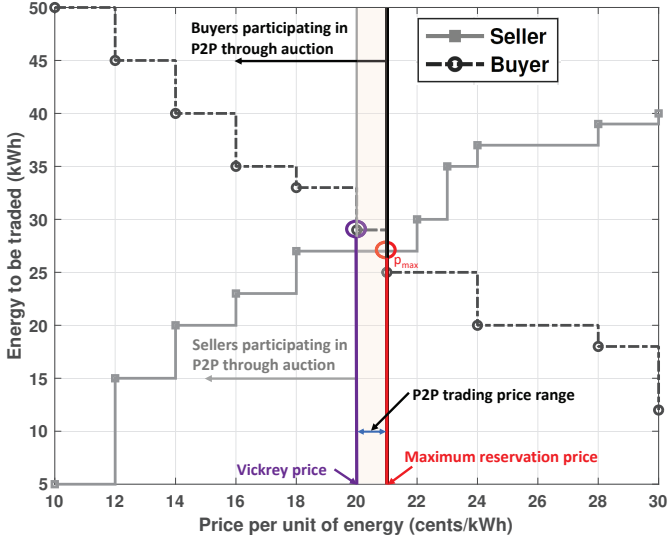}
\caption{A demonstration of how the auction price is determined in the proposed scheme.}
\label{fig:AuctionPrice}
\end{figure}
Third, the auctioneer generates the aggregated supply and demand curves using \eqref{eqn:step1_eqn1} to determine the intersection point of the two curves, as shown in Fig.~\ref{fig:AuctionPrice}. Clearly, the price $p_\text{max}$ at the intersection point in Fig.~\ref{fig:AuctionPrice} refers to the highest reservation price for the sellers. Now,  once $p_\text{max}$ is determined, there are different mechanisms available in the literature to decide the auction price\footnote{Note that while such dynamic pricing is not being used for residential electricity rates in most parts of the world at present, with the emergence of P2P energy trading, it is being used in existing P2P literature~\cite{Tushar_TSG_May_2016} and pilot trials~\cite{Long_Conf_2017}.}. For example, in the Vickrey auction mechanism \cite{Vickrey_JF_Mar_1961}, the auction price (also known as Vickrey price $p_\text{vic}$) is considered to be the second highest reservation price. In \cite{Tushar_TSG_May_2016}, the authors propose a Stackelberg game and show that the auction price corresponds to the trading price at the Stackelberg equilibrium. Nevertheless, without loss of generality, in this work, the auction price $p_\text{auc}$ is assumed to be as same as $p_\text{max}$, that is, $p_\text{auc} = p_\text{max}$.

Once the auction price $p_\text{auc}$ is calculated, then determine the number of buyers $|\mathcal{B}_a| = B_a \leq B$ and sellers $|\mathcal{S}_a|=S_a\leq S$ that will participate in P2P trading based on the auction price. For example, according to Fig.~\ref{fig:AuctionPrice}, clearly all the prosumers $\forall n\in\mathcal{B}_a\cup\mathcal{S}_a$ that satisfy the condition $\{p_\text{auc}:p_{n,s} \geq p_{n,b}\}$ will trade their energy among themselves using P2P trading. 

\emph{Step 2- Allocation of energy:} Once $p_\text{auc}$ is determined, the amount of energy $e_{n,s}$ that each seller $n\in\{1, 2, \hdots, S_a\}$ sells in the auction market is influenced by the demand of each buyer $m\in\{1, 2, \hdots, B_a\}$. In particular,
\begin{eqnarray}
e_{n,s}= \begin{cases}
E_{n,s} & \text{if} \sum_{n=1}^{S_a}E_{n,s}\leq \sum_{n=1}^{B_a}E_{n,b}\\
\left(E_{n,s} - \eta_n\right)^+ & \text{if} \sum_{n=1}^{S_a}E_{n,s}> \sum_{n=1}^{B_a}E_{n,b}
\end{cases}.
\label{eqn:allocation}
\end{eqnarray}
Clearly, each seller prosumer can sell their reserved energy to the buyers when the deficiency is at least equal to the surplus. If the surplus is higher, however, each seller $n$ experiences a burden of $\eta_n$ that they cannot sell to the buyers. There could be different ways to determine the burden to each seller. Examples of such techniques include proportionate distribution and equal distribution. In this study, we propose to use the equal distribution scheme for deciding the burden of each prosumer. This is due to the fact that equal distribution has been proven to be more suitable to deliver a strategy proof auction mechanism~\cite{Huang_CI_Nov_2002}. To this end, we define $\eta_n$ as
\begin{equation}
\eta_n = \frac{1}{S_a}\left(\sum_{n=1}^{S_a}E_{n,s}-\sum_{n=1}^{B_a}E_{n,b}\right).
\label{eqn:burden}
\end{equation} 
\subsubsection{Formation of different coalition}
\begin{figure}[t!]
\centering
\includegraphics[width=\columnwidth]{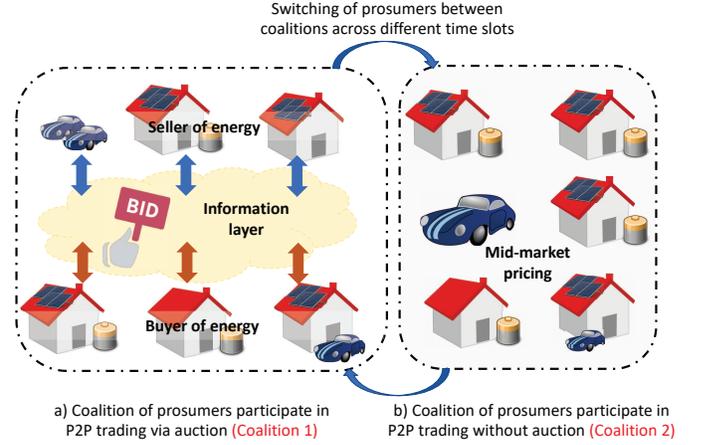}
\caption{This figure demonstrates how a prosumer may choose different coalitions to participate in P2P trading at different time slots of the day.}
\label{fig:CoalitionFormation2}
\end{figure}
As the auction is designed, clearly, prosumers, for which the condition $\{p_\text{auc}:p_{n,s} \geq p_{n,b}\}$ is not satisfied, will not trade their energy using the auction price. Hence, we consider that the prosumers that cannot participate in P2P trading based on $p_\text{auc}$ form another coalition among themselves to trade their energy in a P2P fashion at the market clearing price. In fact, mid-market rate is currently being used in a different P2P pilot project\footnote{For example, in the Horizon P2P-SmartTest community microgrid.}\cite{Long_Conf_2017}. In this work, we define the mid-market price selling $p_\text{mid}$ as a function of $p_\text{auc}$ and feed-in-tariff (FiT) price $p_\text{FiT}$, that is 
\begin{equation}
p_\text{mid} = \frac{p_\text{auc}+p_\text{FiT}}{2},
\label{eqn:mid-market}
\end{equation} and the mid-market buying price as $(1+\beta)p_\text{mid}$. Here, $\beta p_\text{mid},~\beta>0$ is the price that buyer needs to pay as a subscription fee for using the network for P2P trading~\cite{Yu_IoTJ_Dec_2014}. Here, it is important to note the mid-market price always brings lower benefit to both sellers and buyers of the respective coalition compared to trade with $p_\text{auc}$. Further, unlike the proposed auction process, prosumers cannot directly choose the trading price under mid-market rate. Thus, it is reasonable to state that prosumers will always be interested to rebid when instructed by the CPS for P2P trading  in the next available time slot .

\subsubsection{Coalition formation algorithm}The formation of a coalition between different prosumers for P2P trading, as shown in Fig.~\ref{fig:CoalitionFormation2}, is influenced by a prosumer's decision to socially interact with one another for trading its energy in selected time slots (when $p_{g,s} =p_{g,s}^* $), given the fact that it cannot trade its energy with the grid. Therefore, in each time slot, when prosumers are instructed for P2P trading, first, a prosumer decides on an energy amount that it wants to trade in the auction market and put a reservation price for that. Second, based on all the reservation prices and submitted energy amounts, auction price $p_\text{auc}$ is determined by the auctioneer following the process explained in Section~\ref{sec:doubleauction} and subsequently the members of the coalition $\mathcal{B}_a\cup\mathcal{S}_a$ that will trade energy with one another based on $p_\text{auc}$ are confirmed. Finally, once $\mathcal{B}_a\cup\mathcal{S}_a$ is confirmed for the selected time slot, the rest of the prosumers $\mathcal{N}\setminus\mathcal{B}_a\cup\mathcal{S}_a$ decides to form a second coalition to trade their energy among themselves at the mid-market pricing schemes. The detail of the algorithm is shown in Algorithm~\ref{algo:1}. The assumption of an accurate forecast of total demand $E_D(t)$ in Algorithm \ref{algo:1} is motivated by the unprecedented performance improvement demonstrated by the recently proposed state estimation algorithms \cite{Wang_TCNS_2019}.

\begin{algorithm}[t!]
\caption{Algorithm to form a stable coalition structure for P2P trading.}
\label{algo:1}
\begin{algorithmic}[1]
\STATE Set $a_n$ and $b_n$ of each prosumer $n$.
\FOR{Each time $t\in\{1, 2, \hdots, T\}$}
\STATE The CPS accurately forecasts the total demand $E_T(t)$ from prosumers.
\IF{$E_T(t)<E_D(t)$}
\STATE Coalition formation algorithm terminates.
\ELSE
\STATE The CPS sets $p_{g,s}(t) = p_{g,s}^*(t)$ according to \eqref{eqn:Leader-strategy}.
\STATE Prosumers receive $p_{g,s}^*(t)$ from the CPS.
\STATE\hspace{-3mm}\textbf{\emph{Coalition formation algorithm}}
\STATE Each seller prosumer $n\in\mathcal{S}$ submits its bid $\{p_{n,s}(t),E_{n,s}(t)\}$ to the auctioneer.
\STATE Each buyer prosumer $n\in\mathcal{B}$ submits its bid $\{p_{n,b},E_{n,b}\}$ to the auctioneer.
\STATE The auctioneer determines $p_\text{auc}$ and consequently $\mathcal{B}_a$ and $\mathcal{S}_a$ following Section~\ref{sec:doubleauction}.
\STATE The mid-market price $p_\text{mid}$ is calculated via \eqref{eqn:mid-market}.
\STATE The prosumers in $\mathcal{B}_a\cup\mathcal{S}_a$ form coalition $1$ and the prosumers in $\mathcal{N}\setminus(\mathcal{B}_a\cup\mathcal{S}_a)$ form coalition $2$ for P2P trading with $p_\text{auc}$ and $p_\text{mid}$ respectively.
\STATE\hspace{-3mm}\textbf{A stable coalition structure is formed at time slot $t$.}
\ENDIF
\ENDFOR
\end{algorithmic}
\end{algorithm}

Once the coalition structure is formed following Algorithm~\ref{algo:1},  the prosumers trade their energy with one another within each respective coalition. Prosumers can adopt any standard available techniques, such as the technique proposed in \cite{Imran_P2P_GM_2019} or a pool based auction, to choose partners for P2P trading within each coalition. For the coalition with the auction price, the traded energy amount by each prosumers is decided by \eqref{eqn:allocation}. As for the coalition with mid-market price, the P2P trading is completed following the process described in \cite{Tushar_Access_Oct_2018}. 

\begin{remark}Indeed, it is possible that, at any selected time slot of P2P trading, the total deficiency is greater than the total surplus within a coalition and vice versa. In that case, the prosumers with energy deficiency can either reschedule their activities to another time slot or trade the energy from a third party, such as a large neighborhood storage, at a different third party price.\end{remark}

\section{Properties of the Stackelberg Game}\label{lab:Properties}
\subsection{Existence of A Stable CSE}To study the properties of the proposed game,  we note that the strategy chosen by the CPS in \eqref{eqn:Leader-strategy} delivers a unique outcome to the CPS for any value of $b$ that satisfies~\eqref{condition_b}. Thus, to determine whether the proposed $\Gamma$ has a unique and stable CSE, it is sufficient to only investigate whether the coalition structure decided by the prosumers through the proposed auction process is $\mathbb{D}_{hp}$ stable.  

Now, to determine the stability of the proposed coalition structure, first, we define the strategy-proof property of an auction mechanism, which ensures that no prosumer can deviate or cheat from its revealed strategy during the auction process without affecting the strategy of other prosumers. Finally, considering the strategy-proof property of the proposed auction mechanism, we prove the stability of the proposed coalition formation framework, which subsequently proves the existence of the CSE.

\begin{definition}
An auction mechanism is said to be strategy-proof, if the participants reveal their true strategies during the auction process and do not cheat and deviate from their chosen strategies during the trading. 
\end{definition}
\begin{theorem}
The proposed auction mechanism followed by the prosumers to decide on their respective coalition and subsequent P2P trading is strategy-proof.
\label{theorem:1}
\end{theorem}
\begin{proof}
Let us assume that the energy amount that each prosumer $n$ reveal to trade via auction is $E_{n,s}^*$ and $E_{n,b}^*$ for a seller and buyer respectively. Now, according to \eqref{eqn:burden}, if the total available supply and demand are $\displaystyle\sum_{n\in\mathcal{S}_a}E_{n,s}^*$ and $\displaystyle\sum_{n\in\mathcal{B}_a}E_{n,b}^*$ respectively, the burden shared by each participating seller $n\in\mathcal{S}_a$ is
\begin{equation}
\eta_n^* = \frac{1}{S_a}\left(\sum_{n=1}^{S_a}E_{n,s}^*-\sum_{n=1}^{B_a}E_{n,b}^*\right).
\label{eqn:burden1}
\end{equation} 
Now, let us assume that one prosumer $i$, which is a seller, cheats during the energy trading and trades $E_{n,s}^{'}$ instead of $E_{n,s}^*$. Then, the burden at the solution is
\begin{equation}
\eta_n^* = \frac{1}{S_a}\left(\sum_{n\in S_a, n\neq i}E_{n,s}^* + E_{{i\in\mathcal{S}_a},s}^{'}-\sum_{n\in B_a}E_{n,b}^*\right),
\label{eqn:burden1}
\end{equation}
which is impossible. This is due to the fact that, as the scheme is proposed, the burden, which is shared equally by all seller prosumers only possess the value $\eta_n^*$ only if all $\forall n\in\mathcal{S}_a$ stick to $E_{n,s}^*$ for trading, as they revealed during the auction.

Similarly, by considering that one buyer $j\in\mathcal{B}_a$ chooses to buy $E^{'}_{n,b}$ instead of $E^*_{n,b}$ during the trading process, it can be shown that the buyers also need to stick to their revealed strategies and cannot cheat in the proposed auction. Hence, no prosumer would cheat and deviate from their chosen strategies without affecting others, and this subsequently proves the strategy-proof property of the proposed auction scheme.
\end{proof}
\begin{theorem}
At any given time slot, the network structure or partitions resulting from the followers response to the decision made by the CPS is $\mathbb{D}_{hp}$ stable.
\end{theorem}
\begin{proof}
To prove this theorem, first, we note that the prosumers that participate in the P2P trading via auction price are decided by the proposed auction mechanism as discussed in the previous section. Since the auction process is strategy-proof, as proven in Theorem~\ref{theorem:1}, no participating prosumers in the auction would leave the coalition and become a part of another new coalition or act noncooperatively during that particular time slot.

Second, as the scheme is designed, prosumers that cannot participate in the auction form another new coalition among themselves for P2P trading at the mid-market price set by \eqref{eqn:mid-market}. Here it is important to note that, indeed prosumers may also decide to act noncooperatively instead of forming another coalition or form multiple joint coalitions. Nonetheless, it is shown in \cite{Tushar_Access_Oct_2018} that at a mid market pricing for P2P trading, prosumers always benefit more by forming a single coalition with one another instead of acting noncooperatively or forming multiple disjoint coalitions. Therefore, it is reasonable to consider that, as a rational entity, each prosumer in $\mathcal{N}\setminus(\mathcal{S}_a\cup\mathcal{B}_a)$ will always choose to form a coalition among themselves without acting noncooperatively.

Thus, at any given each time slot of P2P energy trading, the coalition formation game among prosumers always results in two coalitions with set of players $(\mathcal{S}_a\cup\mathcal{B}_a)$ and $\mathcal{N}\setminus(\mathcal{S}_a\cup\mathcal{B}_a)$, in which no player has any incentive to leave its own coalition for a greater benefit. This establishes the fact that the network structure or partitions resulting from the followers response to the decision made by the CPS is $\mathbb{D}_{hp}$ stable. 
\end{proof}

Now, based on the above discussion, clearly, in response to the unique decision made by the CPS in \eqref{eqn:Leader-strategy}, the prosumers form a stable group of partitions for P2P trading. Consequently, the following corollary holds true.
\begin{corollary}
The proposed cooperative Stackelberg game $\Gamma$ between the CPS and the prosumers possesses a unique and stable CSE.
\label{cor:corollary1}
\end{corollary}
\subsection{Prosumer-Centric Property}While the proposed $\Gamma$ is proven to be effective for reducing the cost to the CPS to zero at the peak hours through P2P trading, it is also important that the scheme does not compromise the benefit to the prosumers. Prosumer-centric can be defined as the properties of a technology that can potentially motivate people to participate in or accept the technology~\cite{Tushar_Access_Oct_2018}. Motivational psychology is an ideal way to demonstrate this property~\cite{Tushar_AE_Mar_2019}. This is particularly due to the fact that a prosumer might not always be motivated through economic incentive only.  Now, we provide a summary of motivational psychology properties and then show whether the proposed scheme possesses these characteristics. To this end, what follows is a summary of properties that a prosumer-centric scheme needs to satisfy:
\begin{itemize}
\item\emph{Transitivity:} Transitivity refers to the property that establishes the fact that only one decision of the prosumer would maximize its utility or minimize its cost compared to other alternative decisions. As a consequence, as a rational prosumer, it makes the rational choice~\cite{Tversky_Business_1986}. 
\item\emph{Dominance:} If an action chosen by a prosumer is better than another option in one state and at least as good in all other states, the dominant option is chosen by the prosumer~\cite{Tversky_Business_1986}.
\item\emph{Rational-economic:} If an action economically benefits a prosumer, the prosumer is most likely to take that action. That is, monetary incentive is the key motivator for practicing a technology~\cite{Shipworth:2000}.
\item\emph{Positive reinforcement:} The property of positive reinforcement assumes that an equivalent positive outcome of the choice made by a prosumer would motivate the prosumer to choose the same action in the future to fulfill his relevant objective~\cite{Hockenbury:2003}.
\item\emph{Elaboration likelihood:} According to this property, a technique needs to have the capability to communicate its benefit to the prosumers following either a central path or a peripheral path~\cite{Petty:1986}. The central path is assumed to be suitable when an individual cares about the issue and the peripheral path is more appropriate when the issue may be subjected to convey unfavorable thoughts due to the ambiguity of the message.  
\end{itemize}

Now, based on the discussion in the previous section, the proposed P2P energy trading scheme clearly possesses the prosumer-centric property due to following reasons:
\begin{itemize}
\item Once a prosumer receives a signal from the CPS to participate in P2P trading, it can either cooperate with other prosumers in the network or it can continue trade its energy with the CPS at a higher price. Now, as stated in Corollary~\ref{cor:corollary1}, the cooperation between prosumers results in a stable solution. Therefore, no prosumer would have any incentive to choose an alternative strategy other than P2P trading, other than cooperate with one another. This satisfies the transitivity and dominance properties of the scheme.
\item Since the choice of a coalition and the trading prices produce a unique and stable CSE, and due to the fact that the price to trade energy with the grid is very high, clearly trading in P2P energy network benefits the prosumers economically. Consequently, the P2P trading satisfies the rational-economic property.
\item By demonstrating how the proposed P2P energy trading can benefit both the grid by reducing its peak demand and the prosumers by reducing their cost of energy purchase from the grid in each time slot of P2P trading, the proposed scheme establishes a peripheral path to communicate with the prosumers. Further, by exhibiting the same beneficial outcome for both the grid and prosumers in every time slot of P2P trading, the proposed energy trading scheme also validates its positive reinforcement property.  
\end{itemize}
Thus, based on the above discussion, it is reasonable to state the following Corollary.
\begin{corollary}
The proposed P2P energy trading scheme is a prosumer-centric technique.
\end{corollary}
\section{Case Studies}\label{lab:NumericalSimulation}
In this section, we show some results from numerical case studies to demonstrate how the proposed P2P energy trading scheme benefits both the CPS and prosumers at the peak hour. For numerical case studies, we consider a residential network with 12 prosumers. Each prosumer is assumed to have a solar panel of 5 kWp. The energy surplus and deficiency of each prosumer is randomly chosen from the range~[2,9]. However, these values could be different for different generation and consumption patterns of prosumers at different locations. To trade the energy at each time slot\footnote{Each time slot is assumed to have a duration of $30$ minutes~\cite{ZhangChenghua_EA_June_2018}.}, the prosumers choose their bidding from the range~[11,15], which is chosen such that the minimum bidding price is higher than the FiT price 10 cents per kWh and the highest bidding price is lower than the grid's selling price 28 cents per kWh at the off-peak hour. The residential solar data is provided by Redback Technologies, Australia, which is a startup company based in Brisbane that provides smart inverter solutions to households. 
\subsection{Benefit to the CPS}
\begin{figure}[t]
\centering
\includegraphics[width=\columnwidth]{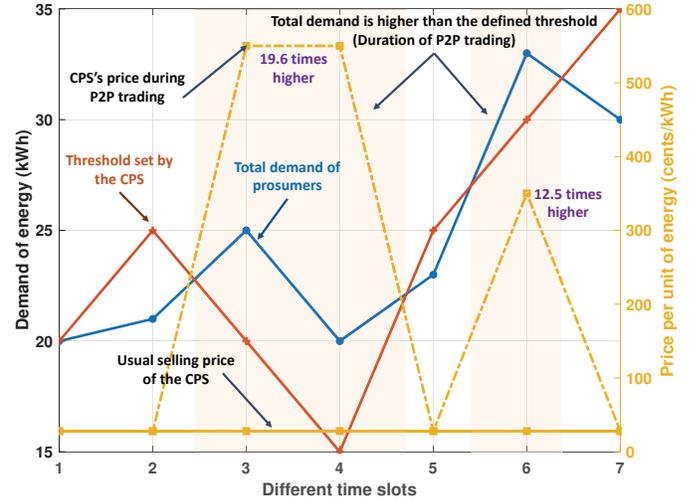}
\caption{This figure demonstrates how the CPS sets selling prices per unit of energy at different time slots based on the total demand from prosumers and the pre-defined threshold, which influence prosumers to participate in P2P trading.}
\label{lab:Demand_price_grid}
\end{figure}
In order to show how the proposed approach benefits the CPS, we first show the strategy adopted by the CPS at different demand scenarios from its contracted prosumers. In particular, in Fig.~\ref{lab:Demand_price_grid}, the selling price per unit of energy chosen by the CPS is demonstrated at different selected time slots during a day. In this figure, we deliberately show only seven time slots for clarity of illustration of how the CPS sets its price based on the demand and threshold at any given time slot. According to Fig.~\ref{lab:Demand_price_grid}, when the total demand of the prosumers is below its threshold of that selected time slot, the grid sells its energy to the prosumers at the general standard off peak rate of $28$ cents per kWh. However, as soon as the total demand becomes greater than the threshold, the CPS changes its trading price for the contracted prosumers to a higher value via \eqref{eqn:Leader-strategy}. For example, at time slot 3 and 4 the price per unit of energy is set to 19.6 times higher than the usual price, whereas at time slot 6, it is 12.5 times. In other instances, when the total demands from the prosumers is within the limits, the CPS sells the energy at the usual rate of 28 cents per kWh. 

\begin{figure}[t]
\centering
\includegraphics[width=\columnwidth]{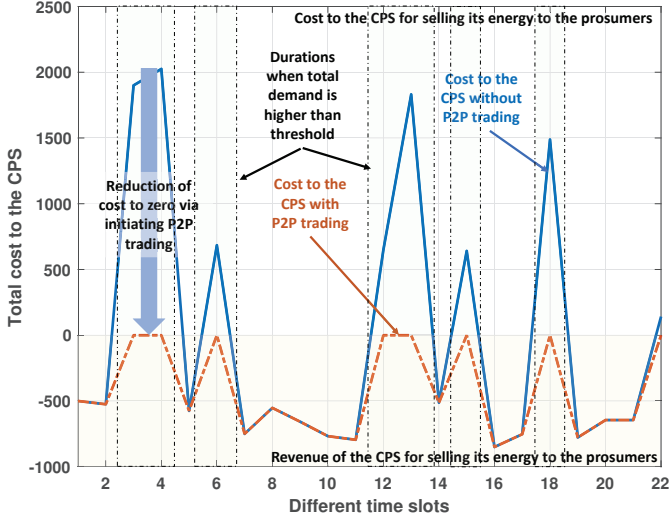}
\caption{The strategy of the CPS reducing its total cost when the total demand of prosumers is higher than the threshold. The negative cost refers to the revenue of the CPS from selling its energy for the case when prosumers' total demand is lower than the threshold.}
\label{lab:Cost_Grid}
\end{figure}
According to the proposed scheme, such a choice of price by the CPS influences the prosumers to participate in P2P trading, instead of trading energy with the CPS. In Fig.~\ref{lab:Cost_Grid}, we show the cost to the CPS at different selected time slots of a day (from 7 am to 5 pm). Based on this figure,
\begin{itemize}
\item At times, when the total demand is lower than the threshold, the CPS is able to make revenue by selling the electricity to the prosumers. For example, at time slots 1, 2, 5, 7--12, 14, 16--18, and 20--22 the revenue to CPS is more than 500 cents.
\item At times, when the total demand is higher than the threshold, due to maintaining reserve and/or running new generation units, the cost to the CPS increases significantly. For instance, at time slot 6, the cost to the CPS increases to around 600 cents for meeting the demand of the prosumers. Similar rises are also observed at time slots 4, 13, 15, and 19.
\item However, for the proposed scheme, the price set by the CPS instructs the prosumers to participate in P2P trading instead of buying the energy from the CPS. This subsequently reduces the cost to the CPS to zero as no excess energy  needs to be generated. 
\end{itemize}
\begin{figure}[t]
\centering
\includegraphics[width=\columnwidth]{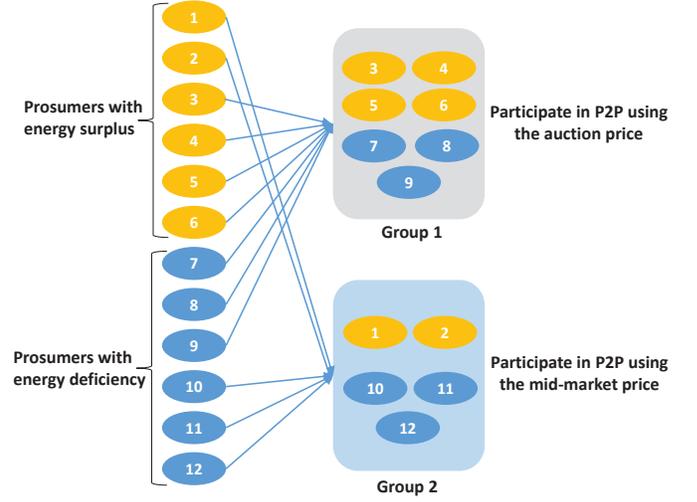}
\caption{Formation of coalitions by different prosumers for P2P energy trading.}
\label{fig:CoalitionFormation}
\end{figure}
\subsection{Benefit to the prosumers}According to the proposed scheme, prosumers participate in P2P energy trading only when they are influenced by the CPS with a very large electricity price signal. Otherwise, they trade their energy with the grid in usual manner. Now, to show how participating in P2P energy scheme can help the prosumers, we show the results for a particular time slot. For example, let us consider time slot 3 of Fig.~\ref{lab:Demand_price_grid} and Fig.~\ref{lab:Cost_Grid}. Clearly, due to the high price of the CPS, prosumers need to participate in P2P trading to balance their demand and surplus. Now, without a loss of generality and as shown in Fig.~\ref{fig:CoalitionFormation}, we consider that prosumers 1, 2, 3, 4, 5 and 6 have surplus, whereas prosumers 7, 8, 9, 10, 11 and 12 need more energy to fulfil their demand. To participate in P2P trading, the prosumers adopt Algorithm~\ref{algo:1} and based on the auction price, prosumers form two coalitions. According to Fig.~\ref{fig:CoalitionFormation}, prosumers 3, 4, 5, 6, 7, 8 and 9 form coalition 1 to trade their energy at the auction price. The rest of the prosumers, that is prosumers 1, 2, 10, 11 and 12, form a second coalition to trade their energy at the mid-market rate, following the process explained in \cite{Tushar_Access_Oct_2018}.
\begin{table}[h!]
\centering
\caption{Demonstration of seller prosumers' increase in their revenues for participating in P2P energy trading.}
\includegraphics[width=\columnwidth]{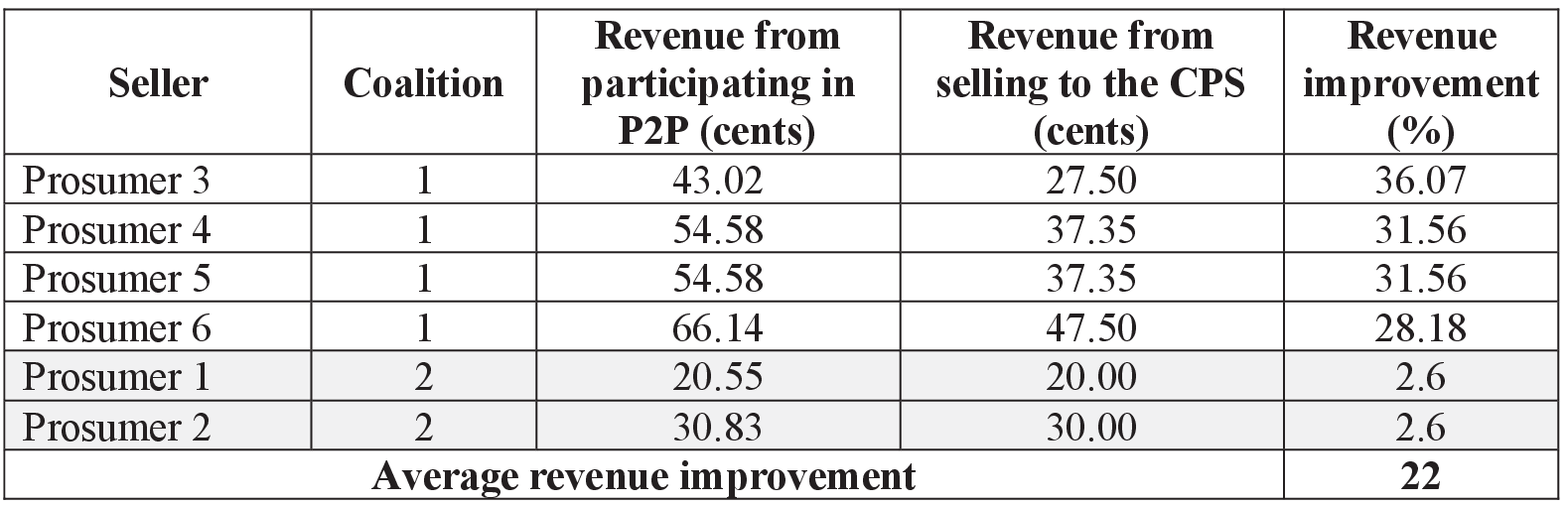}
\label{table:1}
\end{table}

Once the coalition is formed and the trading of energy is taken place at respective coalitions, the sellers and buyers of energy have their own revenue and costs. Now, we show how each prosumer benefits compared to the case of trading energy with the CPS or a third party (such as neighborhood storage) at the selected time slot. In Table~\ref{table:1}, we show the revenue that the seller prosumers in different coalitions obtain for participating in P2P trading. Based on Table~\ref{table:1}, following properties of the scheme can be summarized: 1) auction based P2P demonstrates greater revenue to the prosumer compared to the revenues to the sellers trade their energy at the mid-market rate. This property always encourages prosumers to rebid in the next available time slot of P2P trading in order to get a better benefit; 2) it is always beneficial for the sellers to sell their energy to other prosumers within the P2P network compared to sell them to the CPS; 3) for the considered system and selected time slot, each prosumer achieves a revenue, which is 22\% greater, on average, than the revenue that it could have obtained by selling its surplus to the CPS.
\begin{table}[h!]
\centering
\caption{Demonstration of buyer prosumers' savings in their costs for participating in P2P energy trading.}
\includegraphics[width=\columnwidth]{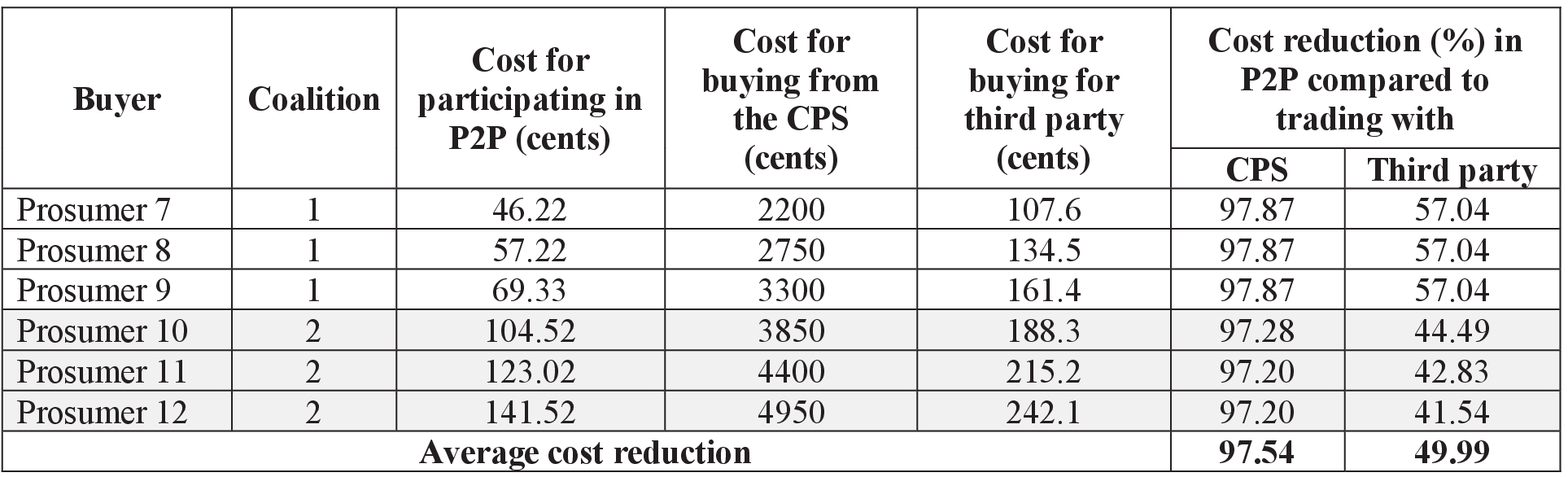}
\label{table:2}
\end{table}

In Table~\ref{table:2}, we show the cost to each buyer within each coalitions for participating in P2P trading. Clearly, P2P trading is the better option for prosumers to meet their demand. Otherwise, the very high price of the CPS\footnote{Example of such a high increase in selling price can also be found in electricity spot markets in high renewable penetrated regions such as South Australia \cite{AEMO_Price_Increase}.} (Fig.~\ref{lab:Demand_price_grid}) increases the cost to each buyer in both coalitions significantly. For example, if  a prosumer decides to buy its required energy from the CPS despite the high price signal, its cost of buying the energy increases by about 97\% compared to buying the energy from the P2P energy trading market. Of course, as a rational entity, a prosumer may also decide to buy its energy from other alternative sources such as third party neighborhood storage~\cite{Mediwaththe_TSG_May_2016}, for example by participating in a non-cooperative game. In such a non-cooperative scenario, each prosumer acts as a rational player who can independently decide on its strategy of whether it wants to buy its energy from the CPS or from the third party in response to the price signal sent by the CPS. Thus, the game framework can be referred to as a non-cooperative Stackelberg game, in which the CPS acts as the leader, which decides on its own strategy of selling price per unit of energy. Prosumers, on the other hand, act as followers, who decide on the venue from which it can buy its energy in response to the price set by the CPS. The solution of the game is a non-cooperative Stackelberg equibrium, at which, due to the very high price of the CPS, all prosumers start trading with the third party. Thus, both the CPS and prosumers reach an equilibrium state from where neither the CPS nor any prosumer want to deviate at the selected time slot.

However, the cost to the prosumers is still very high compared to the price in the P2P trading market. In particular, as can be seen from Table~\ref{table:2}, participating non-cooperatively increases the cost to the prosumer by on average 50\% compared to the P2P energy trading market. Thus, when the CPS sends the high price signal to the prosumers, participating in the proposed P2P energy trading scheme is the best option for both sellers and buyers, instead of acting non-cooperatively, considering their respective revenues and cost.

\begin{table}[h]
\centering
\caption{This table shows how the change in total number of prosumers impacts the total cost to the CPS and prosumers in P2P trading.}
\includegraphics[width=\columnwidth]{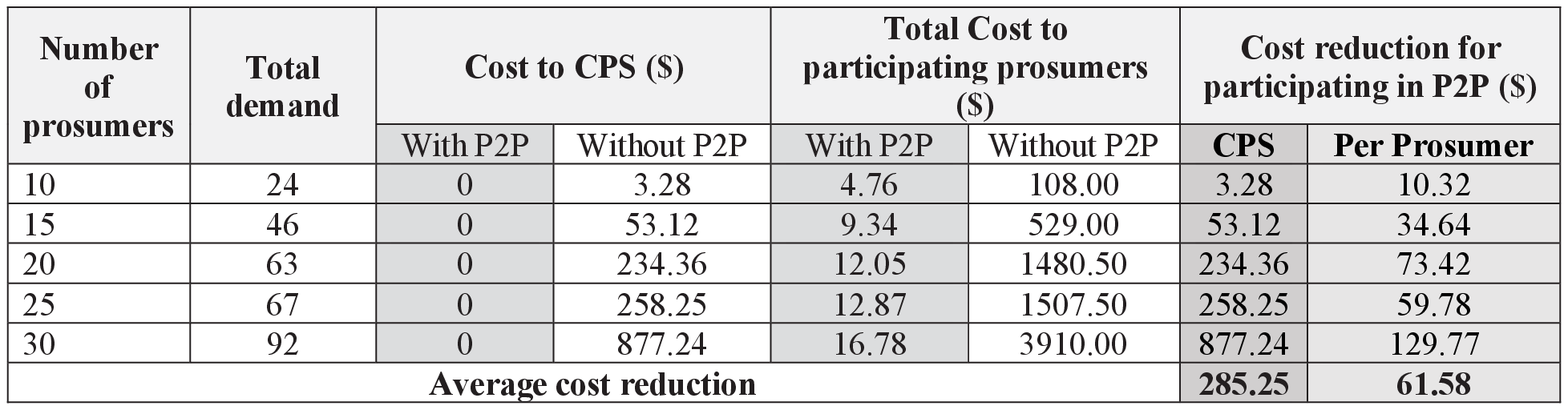}
\label{table:3}
\end{table}
Finally, in Table~\ref{table:3}, we show how an increase in the total number of prosumers impacts the total cost to the CPS and average cost per participating prosumers. Based on the results demonstrated in Table~\ref{table:3},
\begin{itemize}
\item The cost to the CPS is always $0$ when prosumers participate in P2P in response to the CPS's price signal, which is consistent with the proposed scheme. However, if there is no participation in P2P, the cost to the CPS increases substantially as the number of prosumers increases, which is mainly due to the increase in total demand to the CPS.
\item The cost per prosumer\footnote{Total number of prosumer was considered instead of just buyer prosumers to calculate the cost per prosumer of the overall system.} increases as the number of prosumers increases. However, due to the higher selling price per unit of energy, the cost per prosumer is much higher for not participating in P2P compared to the case when prosumers participate in P2P. On average, for all the considered prosumers number, the cost per prosumer shows a \$61.58 reduction for participating in P2P for the considered parameter values in this case study. As for the CPS, the cost savings for the proposed case is \$285.25, on average, when prosumers follow the scheme proposed in this study.
\end{itemize}
Thus, engaging in P2P energy trading by the prosumers following the proposed framework brings benefit to both the CPS and prosumers. Hence, the proposed scheme has the potential to attract the CPS and prosumers to participate in scenarios that are aligned with the assumptions outlined in the paper.

Finally, we note that the computational complexity of the proposed algorithm falls within a category of that of a single leader multiple follower Stackelberg game, which has been shown to be reasonable in numerous studies such as in \cite{Sabita_TSG_2013} and \cite{Wayes-J-TSG:2012}. As such, the computational complexity is feasible for adopting the proposed scheme.
\section{Conclusion}\label{lab:Conclusion} 
In this paper, we have studied the analytical framework of a peer-to-peer energy trading scheme that can help the energy grid to get rid of excess demand from the prosumers at the peak hour, while, at the same time, confirmed a prosumer-centric solution. To this end, we have proposed a cooperative Stackelberg game by assuming the centralized power system as the leader and prosumers as followers. A closed form expression has been presented to capture the decision making process of the leader, whereas we have formulated a double-auction based coalition formation game to facilitate prosumers' decisions of the energy trading parameters. The properties of the proposed cooperative Stackelberg game have been studied and it has been shown that the game possesses a unique and stable Stackelberg equilibrium. Further, we have proposed an algorithm that enables the prosumers and the centralized power station to reach the equilibrium of the game. Finally, we have presented some numerical case studies to show how the proposed scheme can ensure benefits to all participating energy entities in the P2P trading.

A potential extension of the proposed work is the investigation of the impact of uncertainty of prosumers' demand as well as their non-cooperative behavior on the CPS's decision making process. Further, how prosumers can exploit the proposed framework to assist the CPS, and simultaneously participate in P2P energy trading among themselves within the grid-tied setting is another direction that is of interest for further investigation.
\def\baselinestretch{.96}

\end{document}